\documentclass[11pt]{article}
\usepackage[T1]{fontenc}
\usepackage[utf8]{inputenc}
\usepackage[noblocks]{authblk}

\usepackage{geometry}
\usepackage{ifthen}
\usepackage{algpseudocode,algorithm,algorithmicx}

\algrenewcommand\algorithmicrequire{\textbf{Precondition:}}
\algrenewcommand\algorithmicensure{\textbf{Postcondition:}}
\usepackage{graphicx,color}

\usepackage{fancybox}
\usepackage{lipsum}
\usepackage{array,float}
\usepackage{url}
\usepackage{amstext,amssymb,amsmath}
\usepackage{amsthm} 
\usepackage{mathrsfs}
\usepackage{dsfont}
\usepackage{hyperref}
\usepackage{pstricks, pst-tree, pst-node}
\usepackage{enumitem}
\usepackage{mathtools}

\newtheorem{theorem}{Theorem}[section]
\newtheorem{lemma}[theorem]{Lemma}
\newtheorem*{lemma*}{Lemma}

\newtheorem{corollary}[theorem]{Corollary}

\theoremstyle{definition}
\newtheorem{definition}{Definition}[section]

\theoremstyle{definition}
\renewenvironment{proof}[1][Proof]{\begin{trivlist}
\item[\hskip \labelsep {\bfseries #1}]}{\end{trivlist}}
\newenvironment{example}[1][Example]{\begin{trivlist}
\item[\hskip \labelsep {\bfseries #1}]}{\end{trivlist}}

\usepackage{accents}
\usepackage[authoryear,round]{natbib}

\newcommand{\expect}[1]{\text{\bf E}\!\left[#1\right]}

\newcommand{\jobs}{\mathcal{J}}
\newcommand{\alg}{\texttt{ALG}}
\newcommand{\Alg}{\texttt{ONL}}

\newcommand{\W}[2]{W_{#1}(#2)}
\newcommand{\Algi}[1]{\alg_{#1}}
\newcommand{\Algit}[2]{\W{#2}{\Algi{#1}}}

\newcommand{\opt}{\texttt{OPT}}

\newcommand{\reg}{\textsc{Reg}}
\newcommand{\switcher}{\texttt{Expert-ALG}}
\newcommand{\switcherbandit}{\texttt{Bandit-ALG}}

\newcommand{\rit}[2]{r_{#2}^{(#1)}}

\newcommand{\Rit}[2]{R_{#2}^{(#1)}}
\newcommand{\integers}{\mathbb{Z}}

\newcommand{\orderresp}{order respecting\xspace}

\newcommand{\lmax}{\ell_{\textrm{max}}}
\newcommand{\dmax}{d_{\textrm{max}}}
\newcommand{\dlmax}{\dmax + \lmax}
\newcommand{\vmax}{v_{\textrm{max}}}
\newcommand{\rmax}{R}
\newcommand{\fts}{\ensuremath{\texttt{FTS}}}
\newcommand{\ftbs}{\ensuremath{\texttt{FTBS}}}
\newcommand{\coin}[1]{\kappa(#1)}
\newcommand{\heads}{\texttt{Heads}}
\newcommand{\tails}{\texttt{Tails}}

\makeatletter
\newsavebox\myboxA
\newsavebox\myboxB
\newlength\mylenA

\newcommand*\xoverline[2][0.75]{%
    \sbox{\myboxA}{$\m@th#2$}%
    \setbox\myboxB\null
    \ht\myboxB=\ht\myboxA%
    \dp\myboxB=\dp\myboxA%
    \wd\myboxB=#1\wd\myboxA
    \sbox\myboxB{$\m@th\overline{\copy\myboxB}$}
    \setlength\mylenA{\the\wd\myboxA}
    \addtolength\mylenA{-\the\wd\myboxB}%
    \ifdim\wd\myboxB<\wd\myboxA%
       \rlap{\hskip 0.5\mylenA\usebox\myboxB}{\usebox\myboxA}%
    \else
        \hskip -0.5\mylenA\rlap{\usebox\myboxA}{\hskip 0.5\mylenA\usebox\myboxB}%
    \fi}
\usepackage{xspace}
\usepackage{enumitem}

\newenvironment{numberedtheorem}[1]{%
\renewcommand{\thetheorem}{#1}%
\begin{theorem}}{\end{theorem}\addtocounter{theorem}{-1}}

\newenvironment{numberedlemma}[1]{%
\renewcommand{\thetheorem}{#1}%
\begin{lemma}}{\end{lemma}\addtocounter{theorem}{-1}}

\begin{document}
\title{Truth and Regret in Online Scheduling}

\author[*]{Shuchi Chawla}
\affil[*]{University of Wisconsin-Madison~(\texttt{shuchi@cs.wisc.edu})}

\author[**]{Nikhil Devanur}
\affil[**]{Microsoft Research~(\texttt{nikdev@microsoft.com})}
\author[$\dag$]{Janardhan Kulkarni}
\affil[$\dag$]{Microsoft Research~(\texttt{jakul@microsoft.com})}

\author[$\ddag$]{Rad Niazadeh}
\affil[$\ddag$]{Cornell University~(\texttt{rad@cs.cornell.edu})}

\date{}

\maketitle
\begin{abstract}
We consider a scheduling problem where a cloud service provider has
multiple units of a resource available over time. Selfish clients
submit jobs, each with an arrival time, deadline, length, and
value. The service provider's goal is to implement a truthful online
mechanism for scheduling jobs so as to maximize the social welfare of
the schedule. Recent work shows that under a stochastic assumption on
job arrivals, there is a single-parameter family of mechanisms that
achieves near-optimal social welfare. We show that given any such
family of near-optimal online mechanisms, there exists an online
mechanism that in the worst case performs nearly as well as the best
of the given mechanisms. Our mechanism is truthful whenever the
mechanisms in the given family are truthful and prompt, and achieves
optimal (within constant factors) regret.

We model the problem of competing against a family of online
scheduling mechanisms as one of learning from expert advice. A primary
challenge is that any scheduling decisions we make affect not only the
payoff at the current step, but also the resource availability and
payoffs in future steps. Furthermore, switching from one algorithm
(a.k.a. expert) to another in an online fashion is challenging both
because it requires synchronization with the state of the latter
algorithm as well as because it affects the incentive structure of the
algorithms. 

We further show how to adapt our algorithm to a non-clairvoyant
setting where job lengths are unknown until jobs are run to
completion. Once again, in this setting, we obtain truthfulness along
with asymptotically optimal regret (within polylogarithmic factors).

\end{abstract}
\pagebreak

\section{Introduction}


We consider an online mechanism design problem inspired by the
allocation and scheduling of cloud services. A scheduler allocates
scarce resources to jobs arriving over time with the goal of
maximizing economic efficiency or social welfare. The jobs are
submitted by selfish users who can lie about the job's value, length,
arrival, or deadline, so as to obtain a better allocation or pay a
cheaper price. Our goal is to design an online mechanism that is
truthful and obtains good welfare guarantees in the worst case.



There is a vast and rich literature on online mechanism design in
settings where selfish agents participate in the mechanism over
time. What makes the scheduling problem described above interesting is
that even ignoring the users' incentive constraints, there are strong
lower bounds in the worst case for the purely algorithmic problem of
scheduling jobs with deadlines to maximize
welfare. \citet{canetti1998bounding} showed, in particular, that no
online algorithm can achieve a less than polylogarithmic competitive
ratio for this problem in comparison to the hindsight optimal
schedule. On the other hand, \citet{lavi2015online} showed that no
deterministic mechanism that is truthful with respect to all of the
parameters can approximate social welfare better than a factor $T$ in
the worst case, where $T$ is the time horizon, even for unit length
jobs on a single machine. In the face of these strong negative
results, a number of works have considered weakening various aspects
of the model in order to obtain positive results, such as requiring a
slackness condition on the jobs' deadlines~\cite{azar2015truthful},
allowing the algorithm to make tardy
decisions~\cite{hajiaghayi2005online}, satisfying incentive
compatibility with respect to only a few of the jobs'
parameters~\cite{cole2008prompt, azar2011prompt}, etc.

\paragraph{\bf A Bayesian benchmark.}
In this paper we follow an alternate approach of competing against a
benchmark inspired by a stochastic model for job arrivals. For many
online problems, the hindsight optimum is too pessimistic and strong a
benchmark to compete against. A classic example from algorithmic
mechanism design is the digital goods auction for which no mechanism
can compete against the hindsight optimum that obtains the entire
social welfare. But picking the right benchmark to compete against,
namely the optimal posted price, has led to the design of many
beautiful mechanisms with robust and strong revenue guarantees.
\citet{HR-08} advocate a general framework for generating an
appropriate benchmark for such online problems---determine the class
of all mechanisms that are optimal for the problem in an appropriate
stochastic setting; compete against the best of these Bayesian optimal
mechanisms for a worst case instance. 

{\em In this paper we apply the \citeauthor{HR-08} approach to online
  scheduling.  We show that for any given finite class of truthful
  scheduling mechanisms, we can design an online mechanism that is
  competitive against the best of the given mechanisms. Our mechanism
  is truthful with respect to all of the parameters of the jobs'
  types, is computationally efficient, and requires no assumptions on
  the input instance. We achieve asymptotically optimal regret
  guarantees with respect to all of the involved problem parameters.
}

While the \citeauthor{HR-08} agenda has been successfully applied in
mechanism design settings, it has not yet seen much use in algorithmic
settings where the worst case optimization problem is hard but
positive results are known under stochastic assumptions.\footnote{The
  idea of combining several online algorithms into one that is nearly
  as good as the best of them has been explored previously, but in the
  limited context of metrical task systems and adaptive data
  structures. See, e.g., \citet{BB00, BCK02}.}  Our hope is that this
work will spur future work in this direction.

Our work is inspired by the recent work of \citet{CDH+17} that shows
that when jobs are drawn from an i.i.d. distribution in every time
period, a family of simple mechanisms achieves near-optimal social
welfare. \citeauthor{CDH+17}'s mechanism is a simple greedy best-effort
mechanism based on posted prices. The mechanism announces a price
per-unit of resource for each time period into the future. When a job
arrives it gets scheduled in a best-effort FIFO manner in the cheapest
slots that satisfy its requirements. \citeauthor{CDH+17} show that if the
number of resources per time period is large enough, then for any
underlying distribution over job types, there exists a set of prices
such that this posted-pricing-FIFO mechanism achieves a $1-\epsilon$
approximation to expected social welfare. Unfortunately, finding the
best price to offer requires knowing the fine details of the
distribution over job types and solving a large linear program. 

\citeauthor{CDH+17}'s mechanisms are parameterized by a single price.
Machine learning techniques have been succesfully used to tune
parameters of heuristics for a wide variety of problems \cite{xu2008satzilla,hutter2009paramils}.  This
is typically done in a batch setting, where past data is used to find
a good setting of parameters for the heuristic.  For an inherently
online problem such as scheduling, we seek to do the parameter tuning
itself in an online manner.  Can we search for the right parameters as
we are running the heuristic?

We model the online mechanism design problem as a sort of ``learning
from expert advice'' problem where the experts are algorithms in the
set of all Bayesian optimal algorithms. We present a general black-box
reduction from the online scheduling problem to online learning
algorithms for the experts setting. Previous work along these lines
has looked at online settings where the mechanism gets a new instance
of the problem in every time step, and the learning algorithm can
adapt the parameters of the mechanism as more and more instances are
seen. In the scheduling setting the algorithms run on a single
instance of the problem, and we want to achieve low regret over all
Bayesian optimal mechanisms for this specific instance.

\paragraph{\bf Challenges.}
There are several novel challenges that the scheduling context
imposes. First, in our setting, jobs can grab resources for multiple
consecutive units of time. So the decision of scheduling a job in the
current time step can affect the availability of resources in many
future steps. Furthermore it is unclear at what time an algorithm
should receive credit for scheduling a job -- when the job starts, or
when it ends, or throughout its execution? What if the algorithm uses
preemption?

Second, while an online algorithm for the experts problem can cleanly
switch from one expert to another during execution, in the scheduling
setting switching can be tricky. On the one hand, we cannot abandon
jobs that have already been started but not finished by the previous
algorithm. On the other hand, we must try to match the state of the
new algorithm that we are following, so as to obtain the same
reward. 
Most importantly, switching impacts the \emph{incentive properties} of the
underlying mechanisms. Indeed there are several ways in which jobs may
lie to affect their outcomes in the combined mechanism, even if they
cannot in the underlying algorithms. For example, jobs may be able to
influence the time at which a switch happens, or the next algorithm
that is selected. Moreover, a job may be able to benefit indirectly by
influencing the schedule of the algorithm that is used just prior to
the algorithm that determines the job's allocation.

\paragraph{\bf Our techniques.}
We design a switching procedure that overcomes all of these
challenges. Each switch takes a few time steps to finish previously
scheduled jobs, synchronize with the state of the new mechanism, as
well as handle the scheduling of intermediate jobs in a manner so as
to preserve their incentive structure. As a consequence of this
``switching delay'' our algorithm incurs a bounded cost for every
switch. Using this switching procedure in tandem with a reduction to
experts with switching costs allow us to obtain a sublinear regret
guarantee in addition to truthfulness.

Our truthful switching technique is oblivious to the details of the
underlying truthful mechanisms and works as long as the underlying
mechanisms are prompt, that is, they announce the allocation and
payment of a job right at the time of the job's arrival. To our
knowledge this is the first result that combines truthful mechanisms
online into an overall truthful mechanism.



\paragraph{\bf Non-clairvoyance.}
Finally, one challenge faced by schedulers in a real-world setting is
that jobs may not know how long they would take to run. Fortunately,
the benchmark suggested by \citet{CDH+17}, namely the optimal
posted-pricing-FIFO mechanism, continues to obtain near-optimal
welfare---the mechanism does not need to know jobs' lengths in order
to make scheduling decisions, although the existence of a good posted
price assumes that lengths along with other job parameters are drawn
from some i.i.d. distribution. 

Non-clairvoyance poses extra challenges in the online scheduling
setting. We can no longer keep track of different algorithms' rewards
or their states at the time of a switch. Synchronization with an
algorithm's state at the time of a switch is crucial because for some
algorithms, such as posted-price-FIFO in particular, out of sync
execution can cost almost the entire social welfare of the
algorithm. (See, e.g., Example~\ref{example:syncing}.)\footnote{ 
We also show that in a continuous time setting no algorithm can get a sub-linear regret, in Appendix~\ref{sec:app-proofs}.} 
We therefore
consider a slightly modified benchmark where the performance of an
algorithm is measured according to the welfare it accumulates if it is
periodically (randomly) restarted with an empty state. These random
restarts do not significantly affect the performance of algorithms
such as posted-price-FIFO in the stochastic setting.

We give a reduction from the non-clairvoyant setting to a multi-armed
bandit (MAB) problem. The main challenge in this reduction is to
couple the random restarts of the expert algorithms with the times at
which the MAB algorithm decides to switch experts, in a manner that
ensures that the restarts are independent of the internal coin flips
of the MAB algorithm. To do so, we partition the scheduling process
into mini-batches synchronized with random restarts, and run the MAB
algorithm over this batched instance. Because switching between
algorithms happens exactly at the time of a random restart, it becomes
possible for us to sync with the states of the expert algorithms.

Random restarts can once again break the incentive properties of the
overall mechanism. We need to take care to ensure that jobs that are
caught in the middle of a restart cannot benefit by misreporting their
arrival or deadline. This necessitates a careful redesign of the
switching protocol for non-clairvoyant settings. As in the clairvoyant
setting, we obtain the optimal dependence within polylogarithmic
factors of the regret on the time horizon.

\paragraph{Outline.} In Section~\ref{sec:clair} we present a truthful online
algorithm for the clairvoyant setting along with an upper
bound on its regret. We extend both the truthfulness and regret
guarantees to the non-clairvoyant setting in
Section~\ref{sec:non-clair}. Section~\ref{sec:lower-bounds} presents matching lower bounds on regret.



\section{Model and definitions}

\subsection{The online job scheduling problem}

An instance of the online job scheduling problem consists of a finite set of jobs $\jobs$, a time horizon $T$, and the number $m$ of resources (machines) available per unit of time.
Each job $j\in\jobs$ arrives at time $a_j\in[T]$, has a deadline $d_j\in[a_j,a_j+\dmax]$ and a processing length $l_j\in[0,\lmax]$. Assume all the deadlines are in the time horizon $[T]$. Completing each time unit of job $j\in\jobs$ generates a value-per-length $v_j\in[0,\vmax]$. 

\paragraph{Online scheduling algorithms.} An online scheduler is an algorithm that determines which jobs to schedule and when, and how much to charge each scheduled job. Each job $j$, at the time of its arrival, reports its arrival time, deadline, length, and value; this four-tuple is called the job's type. The scheduling algorithm determines whether or not to schedule the job (\emph{admission control step}) and, if the job is scheduled, maps it to a set $\tau_j$ of time units (\emph{scheduling step}) and charges it a payment $p_j$. If $|\tau_j\cap [a_j,d_j]|\ge l_j$, that is, the job is allocated at least $l_j$ time units before its deadline, then the job obtains a utility of $(v_j-p_j)\cdot l_j$. The schedule produced by the algorithm is feasible if no more than $m$ jobs are assigned to each time unit. 

We now discuss various features of online scheduling algorithms:
\begin{itemize}[leftmargin=0pc,itemsep=4pt]
\item[] {\bf Preemption:} We say that the algorithm is {\em non-preemptive} if the set $\tau_j$ consists of contiguous time units for every job $j$. In other words, when a job is started, the algorithm processes it without pausing until it is finished. 
\item[] {\bf Truthfulness:} A scheduling algorithm is {\em truthful} if for every job $j$, fixing the reported types of jobs in $\jobs_{-j}$, job $j$'s utility is maximized by reporting its true type. Jobs can misreport any of the four components of their type, however, following convention we assume that jobs cannot report an earlier arrival time.
\item[] {\bf Promptness:} A scheduling algorithm is {\em prompt} if for every job $j$, the job's allocation and payment, $(\tau_j, p_j)$, are determined at the time of the job's arrival. At times we will refer to a weaker property: an algorithm is {\em order respecting} if for every job $j$, the job's allocation and payment, $(\tau_j, p_j)$, are functions of jobs in $\jobs$ that arrive prior to $j$ and not of those jobs that arrive after $j$.
\item[] {\bf Clairvoyance:} 
The \emph{clairvoyant scheduling} problem is the setting where every job $j$ reports its length $l_j$ to the scheduling mechanism, together with other parts of its type, upon its arrival, whereas in the \emph{non-clairvoyant scheduling problem} jobs do not report their lengths upon arrival. In fact, the scheduling mechanism observes the length of a job only after it completes the job. Since the length of job $j$ is unknown prior to its completion in the non-clairvoyant scheduling, we have to slightly modify other aspects of the setting: 
\begin{itemize}
\item we change the definition of deadline $d_j$ to denote the latest time that $j$ can be started.\footnote{Note the difference with deadlines in the clairvoyant setting, where $d_j$ was defined to be the latest time that $j$ could be completed.}
\item We do not allow preemption in the non-clairvoyant setting. 
\end{itemize} 
With these two modifications, it is indeed guaranteed that if a job $j$ is allocated at a time no later than its deadline, then it will be scheduled properly, i.e. it will be given enough time to be completed. 
\end{itemize}




Let $\Alg$ denote an online scheduling algorithm, and let $J(\Alg) = \jobs\cap\{j: |\tau_j\cap [a_j,d_j]|\ge l_j\}$ denote the set of jobs that receive service in $\Alg$. We use $\W{t}{\Alg,\jobs}$ to denote the value generated by the algorithm at time unit $t$:
\begin{align*}
  \W{t}{\Alg,\jobs} = \sum_{j\in J(\Alg):\, t\in\tau_j} v_j
\end{align*}
The total value generated by the algorithm, a.k.a. its {\em social welfare}, is given by:
\begin{align*}
  \W{}{\Alg,\jobs} = \sum_{t=1}^T \W{t}{\Alg,\jobs}  = \sum_{j\in J(\Alg)} v_j l_j 
\end{align*}
We drop the argument $\jobs$ when it is clear from the context.

\paragraph{Regret minimization in online scheduling.}

We consider an online learning problem, where we are given a finite set of scheduling algorithms and our goal is to compete with the best one in hindsight with respect to the social welfare objective. Let $\{\Algi{1},\ldots,\Algi{n}\}$ be the set of $n$ online schedulers. Given an instance $\jobs$, let $\opt(\jobs) = \max_{i\in [n]} \W{}{\Algi{i},\jobs}$ denote the social welfare obtained by the hindsight optimal algorithm on this instance. Let $\Alg$ denote our online scheduling algorithm. The regret of $\Alg$ is defined as:
\begin{align*}
  \reg(\Alg) \triangleq \max_{\jobs} \left( \opt(\jobs) - \W{}{\Alg, \jobs} \right)
\end{align*}

\subsection{Learning from expert advice}

We will reduce the regret minimization problem for online scheduling to the problem of learning from expert advice. In the latter, we are given $n$ experts indexed by $i$. In each time step $t\in [T]$, the online algorithm must choose a (potentially random) expert, $i_t\in [n]$, to follow. An adversary then reveals a reward vector $\{\rit{i}{t}\}$. We assume that the adversary is oblivious, that is, it cannot observe the internal coin flips of the algorithm. The total payoff of expert $i$ is given by $\sum_{t\in [T]} \rit{i}{t}$. The payoff of the algorithm is given by $\expect{\sum_{t\in [T]} \rit{i_i}{t}}$, where the expectation is taken over the algorithm's internal coin flips. The regret of the algorithm is:
\[
\left( \max_{i\in [n]} \sum_{t\in [T]} \rit{i}{t} \right) - \expect{\sum_{t\in [T]} \rit{i_i}{t}}
\]
Let $\rmax$ denote an upper bound on $\rit{i}{t}$ for any $i\in [n]$ and $t\in [T]$. Then, several different online algorithms are known to achieve a regret of $O(\rmax\sqrt{T\log n})$, and this bound is tight \cite{freund1995desicion,kalai2005efficient,cesa2006prediction}. 

\paragraph{Experts with switching costs.} This is a variant of the problem of learning from expert advice in which the algorithm faces a {\em switching cost} of $C$ units every time it switches from one expert to another in consecutive time steps. In particular, the payoff of the algorithm is given by $\expect{\sum_{t\in [T]} \rit{i_i}{t}} - C \left|\{t\in [T] : i_t\ne i_{t-1} \}\right|$. The first term corresponds to the rewards and the second corresponds to the switching cost. Accordingly, the regret of the algorithm is:
\[
\left( \max_{i\in [n]} \sum_{t\in [T]} \rit{i}{t} \right) - \expect{\sum_{t\in [T]} \rit{i_i}{t}} + C\,\expect{\left|\{t\in [T] : i_t\ne i_{t-1} \}\right|}
\]
\begin{theorem}[\citet{kalai2005efficient}]
  \label{thm:regret-switching}
There is an algorithm $\switcher(C)$ for the experts problem with a switching cost of $C$ such that 
\[  \reg(\switcher(C)) \leq O\left(\sqrt{R(R+C)T\log n}\right) .\]
\end{theorem}

\paragraph{Multi-armed bandit setting.} In the multi-armed bandit (MAB) setting, the online algorithm may only observe the reward $\rit{i_t}{t}$ of the expert that it selects at time $t$, and cannot observe the remaining rewards. Algorithms for MAB typically mix some exploration alongside following the recommendation of an online learning algorithm for the full-information setting. \begin{theorem}[\citet{auer1995gambling}]
	\label{thm:exp3}
	There is an algorithm $\switcherbandit$ for the MAB problem  with 
	\[  \reg(\switcherbandit) \leq O\left(R\sqrt{Tn\log n}\right) .\]
\end{theorem}

We can further consider an extension of the MAB setting to the setting with a switching cost. 
For this problem, we mainly use that there is no algorithm with a regret of $o(T^{2/3})$ \citep{Dekel2014}, to get a similar lower bound for our problem. 
The precise statement of their result is in Section~\ref{sec:lower-bounds}.

\section{The clairvoyant setting}
\label{sec:clair}
In this section we consider the online scheduling problem in the clairvoyant setting, namely where every job reports its length (in addition to the rest of its type) at the time of its arrival. We are given $n$ online scheduling algorithms, $\Algi{1},\ldots,\Algi{n}$, and our goal is to design an online algorithm that minimizes regret relative to the best of the $n$ algorithms in hindsight. 
We begin by showing how to switch between algorithms in a way that preserves truthfulness in Section~\ref{sec:truth}. 
In Section~\ref{sec:red-exp} we present a reduction from this problem to the problem of learning from expert advice with switching costs. In Section~\ref{sec:lower-bounds} we prove that the regret guarantee we obtain from the reduction is optimal. 




\subsection{Truthful switching}
\label{sec:truth}

In this section, we show how to switch between truthful mechanisms while preserving truthfulness and making sure the loss in welfare is bounded. 
We consider the following setting. Let $A$ and $B$ be two \orderresp truthful scheduling mechanisms. 
Our goal is to switch from mechanism $A$ to mechanism $B$ at time 0.
(This is just a normalization of the time index for ease of notation.) 
We consider show how to perform this switch in the clairvoyant
setting, and extend our algorithm to the non-clairvoyant setting in
Section~\ref{sec:nc-truth}. 
The loss in welfare from our switching algorithm is captured in the following lemma. 
\begin{lemma} \label{lem:truthfulswitching}
	Given \orderresp truthful mechanisms $A$ and $B$, 
	there exists an \orderresp truthful mechanism $C$ that obtains welfare at least 
	$$ \sum_{t\le 0} W_t(A) + \sum_{t\geq1} W_t(B) - 2v_{\max} d_{\max}m.$$ 
\end{lemma} 
In particular, all jobs that arrive by time 0 and are completed by mechanism $A$ are also completed by $C$. 
We can compose any number of ``switching'' steps, losing an additive $2v_{\max} d_{\max}m$ amount in welfare each time. 
\begin{theorem} 
	\label{thm:manyswitches} 
Suppose we wish to switch among many \orderresp truthful mechanisms as follows: 
start with $A_0$ at time 1, switch to $A_1$ at time $t_1$, then to $A_2$ at time $t_2$ and so on till you switch to $A_L$ at time $t_L$ for some $L \in \integers_+$. Let $t_0 = 0 $ and $t_{L+1} = T$ for notational convenience. 
Then there is an \orderresp truthful mechanism whose welfare is at least 
\[  \sum_{i=0}^{L} \sum_{t \in (t_i, t_{i+1}]}  W_t(A_i) -  2Lv_{\max}m d_{\max}. \]
\end{theorem} 
\begin{proof}
Let $B_1$ be the mechanism obtained by applying Lemma \ref{lem:truthfulswitching} to switch from $A_0$ to $A_1$ at time $t_1$. 
Apply the lemma again to switch from $B_1$ to $A_2$ at time $t_2$; let the resulting mechanism be $B_2$. 
Continuing this way, we apply the lemma to switch from $B_i$ to $A_{i+1}$ at time $t_{i+1}$,  to get mechanism $B_{i+1}$, for all $i$ up to $L-1$. 
The resulting mechanism at the end, $B_L$, is the one we want. 
\end{proof} 
%

In the rest of this section we prove Lemma \ref{lem:truthfulswitching}. 

\subsubsection{The online switching algorithm}
\begin{definition}
The mechanism  $C$ (Lemma~\ref{lem:truthfulswitching}) is as follows. 
\hrule
\vspace{0.05in}
\begin{enumerate}[itemsep=4pt,topsep=2pt,parsep=0pt,partopsep=0pt,leftmargin=25pt]
\item For jobs that arrive by time $0$, mimic mechanism  $A$ and return the same allocation, schedule  and prices. 
Observe that jobs that are scheduled in this step are terminated by
time $\dmax$.
\item Mark the remaining time slots in $[1,\dmax]$ as \emph{unavailable}. 
This means that for all jobs $j$  that were not considered in the previous step (because $a_j > 0$)  
and have deadline $d_j\le d_{\max}$, we decline service and charge a price of $0$.
\item For all remaining jobs, i.e., jobs $j$ with $a_j>0$ and $d_j> d_{\max}$, consider the jobs in the order of arrival,  and do the following: 
  \begin{enumerate}            
  \item If $B$ rejects $j$, then reject $j$.
	\item     If there are not enough slots available to cover $j$'s length prior to its deadline, reject $j$. 
  \item Otherwise, accept and schedule $j$ in a ``best effort'' manner.
  Specifically, assign to the job
all of the slots that it gets in $B$ and that are still
available in $C$'s schedule. If any of these slots is unavailable,
replace it with the earliest available slot in $C$'s schedule.
We call these newly assigned slots the ``replacement'' slots for job $j$. 
  Charge $j$ the same payment as in mechanism $B$. 
  \end{enumerate}
\end{enumerate}
\hrule
\end{definition}

\paragraph{Design choices.} We explain the design choices made in the above mechanism. 
In step (1) we continue to process jobs that arrive by time 0 according to mechanim $A$. 
If we abruptly stop mechanism $A$, then there may be an incentive for some jobs to lie so that they get scheduled by time 0. 
In step (2) we make the remaining slots unavailable. 
Why not directly go to step (3) and schedule jobs that $B$ has accepted in a best effort manner? 
One of the properties we need for truthfulness to hold is that jobs
that arrive after time 0 finish at a time in mechanism $C$ that is no
earlier than their finish time in mechanism $B$.
Otherwise there may be an incentive for a job to lie so that it gets
accepted in $B$ but is scheduled to finish after its true deadline, 
whereas mechanism $C$ ends up scheduling it within its true deadline. 
Lemma~\ref{lem:replacement} below shows that the algorithm $C$ satisfies this property 
In step (3) (a) if we start considering jobs that $B$ rejected because we have some more available slots than $B$, we might break the truthfulness of $B$. 
Finally, in step (3) (c) we first assign the same slots to the job as
in $B$ in order to ensure the no early completion property. 
Assigning  the remaining available slots in the chronological order
is crucial for the welfare analysis.

\subsubsection{Truthfulness}

We begin by proving the no early completion property.
\begin{lemma} \label{lem:replacement}
	Any replacement slot assigned in step (3) (c) is always later in
        time relative to the unavailable slot it replaces. 
\end{lemma} 
\begin{proof}
	Suppose one of the slots assigned to a job $j$ in mechanism $B$, say at time $t$,  is unavailable in mechanism $C$. 
	If $t \leq \dmax$ then by construction the replacement slot is later. 
	Otherwise, $t$ itself is a replacement slot for some other job $j'$. 
	The arrival time of $j'$, $a_{j'}$ is no larger than $a_j$ because jobs are processed in FIFO order. 
	Since replacement slots are assigned in chronological order, 
	all slots in $[a_{j'},t-1]$ must have been unavailable when $t$ was assigned to $j'$. 
	Now all slots in $[a_j,t]$ are unavailable when we consider job $j$, so its replacement  for slot $t$ can only be later. 	
\end{proof}

\begin{lemma} \label{lem:truthfulness}
Mechanism $C$ is truthful. 
\end{lemma} 
\begin{proof} 
We consider three cases depending on which step of the mechanism  handles the job. 
Recall that we assume that jobs cannot report an earlier arrival time.
\begin{enumerate}[itemsep=4pt,topsep=2pt,parsep=0pt,partopsep=0pt,leftmargin=*]
\item Suppose that $a_j\le 0$, which means that the job gets processed in step (1) and 
gets an allocation and payment as per mechansim $A$.  
If the job reports an arrival time $> 0$, then it is not processed in step (1) and gets none of the slots in 
$[a_j, d_j]$, because $d_j\le d_{\max}$, 
and all those slots are marked unavailable at the beginning of step (2).
Any other misreport means that the job still gets processed in step (1). 
  Now we can appeal to the fact that algorithm $A$ is truthful to assert that the job does not benefit from misreporting its type. 
  
\item Suppose that $a_j>0$ and $d_j\le d_{\max}$, which means it is processed in step (2). 
In this case, regardless of its actual report the job gets no slots in its time window. 

\item Suppose that $a_j > 0 $ and $d_j> d_{\max}$, which means it is processed in step (3). 
Since the job cannot report an earlier arrival time, it cannot be processed in step (1), and reporting 
a deadline $\leq \dmax$ means it gets no slots. 
Hence the only misreports we need to consider are such that the job is still processed in step (3). 

The truthfulness of $B$ should now imply that no misreport can be beneficial in $C$ as well. 
This is almost true since, for instance, the price paid is the same in both (on acceptance). 
However, there is a possibility that misreporting a later deadline in $B$ (possibly combined with a misreport of other parameters) results in a lower price, 
but that in $B$'s schedule the job finishes after its true deadline $d_j$.
This would be a non-beneficial misreport in $B$ but could be beneficial in $C$ if it actually finishes earlier than $d_j$ in $C$, while enjoying the lower price. 
Lemma \ref{lem:replacement} ensures this does not happen.
\end{enumerate}
\end{proof}
  
%
%

\subsubsection{Welfare}
Define a time slot $t> \dmax $ to be ``free'' if mechanism $C$ schedules fewer jobs in time $t$ than mechanism $B$. 
The \emph{number} of free slots at time $t$ is the difference, given
that it is non-negative, and zero otherwise. We first argue that there
are few free slots in $C$'s schedule.

\begin{lemma} \label{lem:freeslot}
	All replacement slots occur before the first free slot. 
\end{lemma} 
\begin{proof} 
	Let $t$ be the first free slot. 
	Consider a job that arrives before $t$. 
	This job is not assigned any replacement slots after $t$ since $t$ is free and hence available, and replacement slots are assigned in chronological order. 
	We will argue that jobs arriving after $t$ have no replacement slots, i.e., they get the same slots as in $B$. 
	This is by induction on the arrival order of these jobs. 
	Consider the very first such job. 
	All earlier jobs arrive before $t$ by definition, and have no replacement slots after $t$ as already argued, 
	therefore all of the slots assigned to this job in $B$'s schedule are available. This is the base case. 
	The argument for the inductive case is almost exactly the same. 	
\end{proof} 

\begin{lemma} \label{lem:nooffreeslots} 
	The total number of free slots is at most $m\dmax$. 
	In particular, if $t$ is the earliest time of a free slot, then all the free slots are in the interval $[t,t+\dmax]$.
		In other words, mechanisms $B$ and $C$ get synchronized after time $t + \dmax$. 
\end{lemma} 
\begin{proof} 
	From Lemma \ref{lem:freeslot}, there are no replacement slots after $t$. 
	Any job that arrives at $t$ or later and is scheduled in $B$ gets the same slots in $C$ as in $B$, and hence there are no free slots corresponding to such a job. 
	All the free slots must correspond to jobs that arrive before $t$, and are therefore 
	in the interval $[t,t+\dmax]$. 
\end{proof}

We are now ready to prove the main lemma of this section. 
\begin{proof}[Proof of Lemma \ref{lem:truthfulswitching}]
	It is easy to see that mechanism $C$ is also \orderresp. 
	Since we already showed that the mechanism is truthful in Lemma \ref{lem:truthfulness}, we only need to argue about the welfare. 
	
	Any job that arrives before time 0 and is accepted by $A$ is also accepted by $C$ and completed, 
	therefore it gets the same welfare as $A$ upto time 0. 
	Now we argue about the total loss in welfare during the time $t \geq 1$. 
	Let $\ell_B(t)$ (resp. $\ell_C(t)$) be the number of jobs scehduled at time $t \geq 1$  by mechanism $B$ (resp. mechanism $C)$, 
	and let $\ell_F(t)$ be the number of free slots at time $t$. 
	By the definition of a free slot, we have that 
	\[ \sum_{t \geq 1} \ell_B(t) \leq \sum_{t \geq \dmax +1}\left(\ell_C (t) + \ell_F(t)\right)   +  m\dmax. \]
	The set of jobs accepted by $C$ is a subset of the set of jobs accepted by $B$, due to steps (2) and (3a). 
	The total length of all jobs that are accepted by $B$ but not by $C$
	is equal to 
	$\sum_{t \geq 1} \left(\ell_B(t) - \ell_C(t)\right)
	\leq \sum_{t \geq \dmax +1} \ell_F(t) + m\dmax \leq 2m\dmax$, 
	where the last inequality is from Lemma \ref{lem:nooffreeslots}. 
	Thus the total loss is at most $2m\vmax\dmax$. 
\end{proof}	

\subsection{Reduction to experts with switching costs }
\label{sec:red-exp}
Let $\switcher(C)$ denote an online algorithm for the problem of learning from expert advice with switching cost $C$ that achieves the regret guarantee of Theorem~\ref{thm:regret-switching}. $\switcher$ is given an instance with $n$ experts, indexed by $i\in [n]$. It specifies for every time step $t\in [T]$ a random expert $i_t$, and then receives a reward vector $\{\rit{i}{t}\}$. Our online scheduling algorithm, that we call \emph{Follow-The-Switcher} or \fts, simulates $\switcher$ in a black-box fashion and follows its advice on which expert, a.k.a. algorithm, to run at every time step.


\begin{definition}
\label{def:FTS}
Given the $n$ online scheduling algorithms, $\Algi{1},\ldots,\Algi{n}$, the Follow-The-Switcher, a.k.a. \fts, algorithm simulates the online algorithm $\switcher(C)$ with $C$ set to $2\vmax\dmax m$. It then proceeds as follows.

\vspace{0.2in}
\hrule
\vspace{0.05in}
\item [] At each time $t\in[T]$:
\begin{enumerate}[itemsep=4pt,topsep=2pt,parsep=0pt,partopsep=0pt,leftmargin=25pt]
\item Simulate algorithms $\Algi{1},\ldots,\Algi{n}$ on the freshly arrived set of jobs.
\item Query $\switcher$ to obtain the index $i_t\in [n]$.
\item If $i_t \neq i_{t-1}$, then switch from $\Algi {i_{t-1}} $ to $\Algi {i_t}$ as described in Section~\ref{sec:truth}. Otherwise continue running the same algorithm $\Algi {i_{t-1}} = \Algi {i_t}$. 
\item Set $\rit{i}{t}\leftarrow\W{t}{\Algi{i_t}}$ for all $i\in[n]$. Send the reward vector $\{\rit{i}{t}\}$ to $\switcher$.
\end{enumerate}
\hrule
\end{definition}

The following theorem now immediately follows from Theorem~\ref{thm:manyswitches}. 
\begin{theorem} 
	\label{lem:fts-regret}
	Let $C=2\vmax\dmax m$. Then the Follow-The-Switcher (\fts) algorithm, described in Definition~\ref{def:FTS}, admits the following regret-bound:
	\begin{align*}
		\reg(\fts)\leq \reg(\switcher(C))\leq O\left(m\vmax\sqrt{T\dmax\log n}\right)~.
	\end{align*}
\end{theorem}

\section{Non-clairvoyant setting}
\label{sec:non-clair}
In this section, we look at regret minimization in the non-clairvoyant setting; 
we recall the main differences here. 
 Every job reports all parts of its type except its length, and the algorithm only observes 
 the length of a job when (and if) it is completed. 
Hence a non-clairvoyant algorithm cannot plan for a complete schedule ahead of time. 
The algorithm maintains a queue of unfinished jobs and 
at every time $t$ decides which job to schedule from this queue at that time.
The deadline for a job is now the number of time slots that the job is willing to wait out. 
If a job passes its deadline, which means that the number of time slots that a job $j$ waits since its arrival exceeds a threshold $d_j$, then the job is deleted from the queue. 
Due to this reason, the notion of promptness is not quite applicable to the non-clairvoyant setting. 
In its place, we use the  \orderresp property, which states that only jobs arriving earlier can influence the allocation and payments for a given job.
We assume that there is a total order on the arrival time of the jobs, by breaking ties arbitrarily in case multiple jobs arrive at the same time.

Similar to the clairvoyant setting, 
we have a set of $n$ online scheduling algorithms, $\Algi{1},\ldots,\Algi{n}$, and we aim to design an online algorithm that minimizes the regret relative to the best of these algorithms in hindsight. 
We begin our discussion in Section~\ref{sec:random-restart}, 
where we show that this benchmark by itself is impossible to compete with,  
which motivates a reasonable modification. 
In Section~\ref{sec:nc-truth} we show how to switch between two mechanisms truthfully. 
Finally in Section~\ref{sec:reduction-to-bandits}, we show how to use multi-armed bandit algorithms to get tight regret bounds. 
\subsection{Scheduling algorithms with random restarts}
\label{sec:random-restart}
\paragraph{Robustness of the benchmark.} In the non-clairvoyant setting, it is easy to come up with examples showing that the welfare obtained by an online scheduling mechanism is very sensitive to \emph{timing} in the adversarial model of jobs, i.e., by slightly changing the starting time of the mechanism the obtained welfare can be drastically different.  This has been demonstrated in Example~\ref{example:syncing}. 
\begin{example} 
	\label{example:syncing}
	Consider running FIFO scheduling with pricing admission control at $p=1$. Suppose we have three jobs $J_1, J_2,J_3$, with $v_1=v_2=v_3=1$. Suppose $(a_1,l_1)=(1,3)$, $(a_2,l_2)=(3,3)$ and $(a_3,l_3)=(4,T-4)$. (All of them have immediate deadlines, which means they need to be scheduled when they arrive.) Normally, we schedule jobs $J_1$ and $J_3$ and generate a welfare equal to $T$. Now consider  starting at time $t=2$. Then we only schedule job $J_2$ (and in the non-clairvoyant setting we will not even notice how valuable job $J_3$ was!) and get only welfare equal to $3$.
\end{example}

As it is clear from this example, the welfare obtained form such a mechanism cannot be a reasonable benchmark for our regret minimization, as one has to think very carefully about when to start running such a mechanism to calculate the benchmark. 
Otherwise, the benchmark mechanism could easily get tricked into a false start. In other words, we need to define the benchmark in a way that is \emph{robust} to this sort of timing issues, independent of the choice of scheduling mechanism defining the benchmark.

\paragraph{Syncing issues.} In the non-clairvoyant setting, our scheduling mechanism is not able to \emph{simulate} an arbitrary candidate scheduling mechanism $\Algi{i}$ starting from an arbitrary time, 
since there is no way of knowing its state. 
Accordingly, following the decisions of a bandit algorithm, similar to what we did in Section~\ref{sec:red-exp},  is generally not possible in the non-clairvoyant setting. 
  
  However, if both our mechanism and the new switched scheduling mechanism restart from a fresh state at exactly the same time after switching (e.g. slightly after the switching time when our mechanism is done with its current jobs) then our mechanism can sync with the new scheduling mechanism.


To address above concerns and to be able to design a truthful online scheduling mechanism that achieves a meaningful regret bound, we introduce a couple of new ingredients in our model and redefine our benchmark. We start by defining the notion of a \emph{random restart} formally as following.
An important property of the way we restart is that it preserves truthfulness: a mechanism that was truthful to begin with is still truthful with a restart. 
\begin{definition} 
\label{def:random-restart}
Given an online scheduling mechanism $A$, we define the \emph{restart} at time $t$ as follows.   
\begin{itemize}
\item During $[t:t+(\lmax+\dmax)]$, mechanism $\mathcal{M}$ continues working on the jobs that have arrived before $t$.
\item If a job $j$ arrives during $[t:t+(\lmax+\dmax)]$, modify it as follows. 
\begin{itemize}
\item Shift its arrival  time to the end of this interval, i.e. $a_j\leftarrow t+(\lmax+\dmax)+1$.
\item Adjust the deadline of the job so that it reflects the time lost during the interval $[t,t+\dlmax]$. 
This might mean some jobs are past their deadline. These jobs are rejected. 
\item Preserve the arrival order. Use the tie breaking rule to make sure the arrival order of jobs whose starting time was set to $t+\dlmax$ is the same as in the original instance. 
\end{itemize}
\end{itemize}
\end{definition}
Having the formal definition of a restart, we ask the following question: how can one define a robust benchmark in the non-clairvoyant scheduling problem, given a set of candidate scheduling mechanisms? Here is an adaptation of our previous benchmark, i.e. welfare of the best-in-hindsight scheduling mechanism, for the non-clairvoyant setting.

\begin{definition}
\label{def:robust-bench} Given candidate scheduling mechanisms $\Algi{1},\ldots,\Algi{n}$,
and a parameter $\gamma\in[0,1]$,  the \emph{random-restarting} mechanisms $\xoverline{\Algi{1}},\ldots,\xoverline{\Algi{n}}$ are defined to be the original candidate mechanisms accompanied by independent random restarts of probability $\gamma$  at every time $t\in[T]$. We define the \emph{random-restarting benchmark} for an instance $\jobs$ to be $\xoverline{\opt}(\jobs) = \max_{i\in [n]} \expect{\W{}{\xoverline{\Algi{i}},\jobs}}$.
\end{definition}

The benefits of using benchmark $\xoverline{\opt}(\jobs)$ are twofold. First, this benchmark is robust to timing issues, because the scheduling mechanism generating the benchmark restarts independently at random at every time $t$ with probability $\gamma$. Therefore, the benchmark loses no more than the generated welfare between two consecutive random-restarts due to timing issues. Second, while this property does not hold in general, for a large family of scheduling mechanisms (such as posted-pricing-FIFO)  and under the \emph{stochastic model} of jobs (e.g. see~\cite{CDH+17}), the welfare loss due to independent (but infrequent) random restarts will easily be bounded. This property of the pair (stochastic model, online scheduling mechanisms), which we call \emph{robustness-to-welfare-loss}, is formalized as following. 

\begin{definition}
\label{def:robust-to-restart}
 Given a distribution over jobs $\mathcal{D}$, the random-restarting benchmark $\xoverline{\opt}(\jobs)$ with parameter $\gamma$ is robust-to-welfare-loss in expectation over stochastic jobs $\mathcal{D}$ if 
\[
\text{\bf E}_{\jobs\sim\mathcal{D}}\!\left[{\xoverline{\opt}(\jobs)}\right]\geq \text{\bf E}_{\jobs\sim\mathcal{D}}\!\left[{{\opt}(\jobs)} \right] -\gamma\cdot T\cdot\vmax(\lmax+\dmax)
\] 
Moreover, a random-restarting online scheduling mechanism $\xoverline{\Algi{i}}$ with parameter $\gamma$ is robust-to-welfare-loss in expectation under stochastic jobs $\mathcal{D}$ if
\[
\text{\bf E}_{\jobs\sim\mathcal{D}}\!\left[\sum_{t\in[T]}{\W{t}{\xoverline{\Algi{i}},\jobs}}\right]\geq \text{\bf E}_{\jobs\sim\mathcal{D}}\!\left[\sum_{t\in[T]}{\W{t}{{\Algi{i}},\jobs}}\right] -\gamma\cdot T\cdot\vmax(\lmax+\dmax)
\]
\end{definition}

Clearly, if all of the random-restarting mechanisms $\Algi{1},\ldots,\Algi{n}$ are robust-to-welfare-loss under stochastic job model $\mathcal{D}$, then  $\xoverline{\opt}(\jobs)$ will also be robust-to-welfare-loss under $\mathcal{D}$.



%

\subsection{Truthful switching in the non-clairvoyant setting}
\label{sec:nc-truth}
We now show how to switch between truthful non-clairvoyant mechanisms. 
We reiterate a subtle aspect of truthfulness in the non-clairvoyant setting: a job that tries to influence the mechanism by switching its position in the arrival order can be treated the same way 
as a job reporting a later arrival time: it cannot be beneficial to do this given that the mechanism is truthful. 
Posted-pricing-FIFO is an example of a mechanism that is both truthful and \orderresp: 
all jobs whose values are less than a threshold price are rejected, and the rest of the jobs are scheduled in arrival order. 

We claim that the random restart algorithm of Definition \ref{def:random-restart} preserves the truthfulness of the underlying scheduling mechanism. The following lemma is proved in Appendix~\ref{sec:app-proofs-ub}.
\begin{lemma}\label{lem:truth-restart}
	{If $A$ is an order respecting truthful mechanism, then $A$ with (arbitrary) restarts is also order respecting and truthful. }
\end{lemma}
	

Combining ideas from the truthful switching algorithm in the clairvoyant setting and the truthful random restart algorithm, we develop a truthful switching algorithm for non-clairvoyant settings that switches from a mechanism $A$ to a mechanism $B$ at time $0$.
\begin{definition}
\label{def:nclr-truthful-switcher}
The mechanism  $C$ is as follows. 
\begin{enumerate}
	\item For jobs that arrive by time $0$, mimic mechanism  $A$ and return the same allocation, schedule  and prices as $A$. 
	All these jobs are completed by time $\dlmax$. 
	\item For all remaining jobs, i.e., jobs $j$ with $a_j\geq 1$ run mechanism $B$ on these with the following modifications. 
	\begin{enumerate}            
		\item If the arrival time of a job is $< \dlmax$, set its arrival time to $\dlmax+1$. 
		\item Adjust the deadline of the job so that it reflects the time lost during the interval $[1,\dlmax]$. 
		This might mean some jobs are past their deadline. These jobs are rejected. 
		\item Preserve the arrival order. Use the tie breaking rule to make sure the arrival order of jobs whose starting time was set to $\dlmax$ is the same as in the original instance. 
	\end{enumerate}
\end{enumerate}
\end{definition}

\begin{lemma}\label{lem:state-match}
	The state of the algorithm $C$ at time $\dlmax$ is the same as the state of the algorithm $B$ at time $\dlmax$, given that $B$ is restarted during the interval $[1,\dlmax]$.
\end{lemma}
\begin{proof} 
	This follows by observing that Step (2) of mechanism 	$C$ is identical to the modifications made during a restart. 
	Any job that arrives by time 0 does not influence the state of mechanism $B$ at time $\dlmax$ in either case. 
\end{proof}

We obtain the following theorem (see Appendix~\ref{sec:app-proofs-ub} for a proof).
\begin{theorem} \label{thm:truthfulswitching}
	Given \orderresp truthful mechanisms $A$ and $B$ in the non-clairvoyant setting, 
	switching mechanism $C$ in Definition~\ref{def:nclr-truthful-switcher} is \orderresp, truthful, and obtains welfare at least 
	$$ \textstyle \sum_{t\le \dlmax} W_t(A) + \sum_{t\geq1 + \dlmax} W_t(B) ,$$ 
	given that $A$ and $B$ are restarted at time 1. 
\end{theorem} 

\subsection{Reduction to multi-armed bandits}
\label{sec:reduction-to-bandits}
In this section, we show how to design a truthful online learning algorithm that minimizes the regret relative to the random-restarting benchmark, i.e.  best-in-hindsight of random-restarting truthful mechanisms $\xoverline{\Algi{1}},\ldots,\xoverline{\Algi{n}}$. Similar to Section~\ref{sec:red-exp}, we consider a relevant adversarial Multi-Armed Bandits problem or MAB, where we have an arm for each of the $n$ online scheduling mechanisms and we have bandit feedback, meaning that any algorithm only observes the reward of the arm it plays and not the other arms. 

In such a setting, we assume having query access to a MAB algorithm $\switcherbandit$ that admits the optimal regret bound $O(R\sqrt{T n \log n})$ in Theorem~\ref{thm:exp3}, where $n$ is the number of arms, $T$ is the time horizon and $R$ is an upper-bound on the reward of an arm. Moreover, for the sake of simplicity, we assume $\switcherbandit$ does not need to know the time horizon $T$ or rewards range $R$ in advance, and it only needs to know these quantities are bounded. This assumption can be removed by using a \emph{doubling trick}: given black-box access to a bandit algorithm $\switcherbandit_1(R,T)$ that requires knowing $R$ and $T$, one can come up with another black-box algorithm $\switcherbandit_2$ with the same asymptotic regret bound that does not need this information. This reduction is explained in Appendix~\ref{appendix:doubling-trick}.

 Our proposed algorithm, which we call \emph{Follow-The-Bandit-Switcher} or \emph{FTBS},  uses $\switcherbandit$ in a black-box fashion when looped in with the right rewards. It then follows $\switcherbandit$'s advice  by truthful switching between arms, as suggested by Theorem~\ref{thm:truthfulswitching}. This helps the FTBS to aggregate truth and welfare guarantees of mechanisms $\xoverline{\Algi{1}},\ldots,\xoverline{\Algi{n}}$ even in the non-clairvoyant setting. We formally define the FTBS as following.

\begin{definition}
\label{def:FTBS}
Given the $n$ online scheduling mechanisms, $\Algi{1},\ldots,\Algi{n}$,  query access to online bandit algorithm $\switcherbandit$, and parameter $\gamma\in[0,1]$,  the Follow-The-Bandit-Switcher mechanism proceeds as follows.

\vspace{0.2in}
\hrule
\vspace{0.05in}
 \item[] Initialize $\hat{t}=1$. \texttt{\footnotesize{(This index is counter for number of \heads)}}
 \item[]Query $\switcherbandit$ for the initial arm $i_1$, and run algorithm $\Algi{i_{1}}$.
 
  \item[] At each time $t\in[T]$:
\begin{enumerate}[itemsep=4pt,topsep=2pt,parsep=0pt,partopsep=0pt,leftmargin=25pt]

	\item  Flip an independent coin $\coin{t}$ with $\textrm{Pr}(\coin{t} = \heads) = \gamma$. 
	\item If coin $\coin{t}$ shows a $\heads$,
	\begin{enumerate}
	\item Let $t'=\max\{t''<t: \coin{t''} = \heads\}$. If no such an integer exists, let $t'=1$.
	\item Update the sequence of bandit rewards between two consecutive $\heads$:
\[
 \rit{i_{\hat{t}}}{x}= 
\begin{cases}
    0,& \text{if }~x\in[t',\min(t'+\lmax+\dmax,t-1)]\\        \W{x}{\ftbs},         & \text{if }~x\in[\min(t'+\lmax+\dmax,t-1)+1, t-1]
\end{cases}
\]	
	    \item Set $\Rit{i_{t'}}{\hat{t}}\leftarrow \sum_{x\in[t',t-1]}\rit{i_{t'}}{x}$ and send this bandit feedback to $\switcherbandit$.
	    \item Set $\hat{t}\leftarrow\hat{t}+1$.
		\item Let the new arm chosen by $\switcherbandit$ be $i_{\hat{t}}\in[n].$
		\item If $i_{\hat{t}} \neq i_{\hat{t}-1}$, switch from $\Algi{i_{\hat{t}-1}}$to ${\Algi{i_{\hat{t}}}}$ using the mechanism in Section~\ref{sec:nc-truth}. 
		\item Otherwise, restart ${\Algi{i_t}}$ at time $t$. 
	\end{enumerate}
	\item If coin $\coin{t}$ shows a $\tails$, continue running ${\Algi{i_{\hat{t}}}}$.
\end{enumerate}

\hrule
\end{definition}

We now state and prove a tight regret bound (up to logarithmic factor) for FTBS.


\begin{theorem} 
\label{thm:regret-nclr}
 The Follow-The-Bandit-Switcher (FTBS) scheduling mechanism, described in Definition~\ref{def:FTBS}, admits the following regret-bound if $\gamma=(\lmax+\dmax)^{-2/3} T^{-1/3}(n \log(n))^{1/3}$:
\[
\reg(\ftbs)\leq O(m\cdot\vmax (\lmax+\dmax)^{1/3} (n\log(n))^{1/3} T^{2/3}\log T)=\tilde{O}(T^{2/3})
\]
where $\reg(\texttt{FTBS})$ is the regret relative to $\xoverline{\opt}$, i.e. the random-restarting benchmark as in Definition~\ref{def:robust-bench}.
\end{theorem}

We finally consider stochastic jobs, and we focus on benchmarks that are robust-to-welfare-loss in expectation under this stochastic model, as described in Definition~\ref{def:robust-to-restart}. The following corollary is immediate.

\begin{corollary}
\label{cor:nclr-stochastic-regret}
Given a distribution $\mathcal{D}$ over jobs and a robust-to-welfare-loss benchmark $\opt(\jobs)$ (Definition~\ref{def:robust-to-restart}), there exists a scheduling mechanism whose expected regret relative to $\text{\bf E}_{\jobs\sim\mathcal{D}}\!\left[\opt(\jobs)\right]$ is bounded by $\tilde{O}(T^{2/3})$. 
\end{corollary}

\begin{proof}[Proof of Theorem~\ref{thm:regret-nclr}]
Let $C_1\triangleq \vmax(\lmax+\dmax)m$ and $C_2\triangleq \vmax m$ be constants. Fix a sequence of coins $\mathbf{K}[T] \triangleq [\coin{1},\ldots,\coin{T}]$ and let $t_1,t_2,\ldots,t_{\hat{T}}$ be the times $t\in[T]$ that coin $\coin{t}$ shows a $\heads$. As a convention, let $t_0\triangleq 1$ and $t_{\hat{T}+1}\triangleq \tilde{T}\geq T$ be the next time that coin $\coin{t}$ flips a $\heads$ if we continue flipping after $T$. By abuse of notation, we will use $\xoverline{\Algi{i}}$ to denote mechanism $\Algi i$ restarted at exactly these times. For each $i\in[n]$ and  $x\in [1,\hat{T}]$, let $\Rit{i}{x}$ be:
\[
\Rit{i}{x}\triangleq\sum_{t=\min(t_{x-1}+\lmax+\dmax,t_{x}-1)+1}^{t_{x}-1} \W{t}{\xoverline{\Algi{i}}}\geq \sum_{t=t_{x-1}}^{t_{x}-1} \W{t}{\xoverline{\Algi{i}}} -C_1
\]
Also, for the last interval $[t_{\hat{T}},T]$ and for each $i\in[n]$, let $\Rit{i}{\hat{T}+1}$  be
\[
\Rit{i}{\hat{T}+1}\triangleq\sum_{t=\min(t_{\hat{T}}+\lmax+\dmax,t_{x}-1)+1}^{T} \W{t}{\xoverline{\Algi{i}}}\geq \sum_{t=t_{\hat{T}}}^{T} \W{t}{\xoverline{\Algi{i}}} -C_1
\]

Therefore, by summing over all intervals $\{[t_{x-1},t_{x}]\}_{x\in[\hat{T}]}\cup [t_{\hat{T}},T]$, we have:
\begin{equation}
\label{eq:nclr-upper}
\forall i\in[n]:~~\sum_{x\in[\hat{T}+1]}\Rit{i}{x}\geq \W{}{\xoverline{\Algi{i}}}-C_1(\hat{T}+1)
\end{equation}
Note that after the coin shows a $\heads$ at time $t_{x}$, our mechanism either stays with the same algorithm and does a restart or the truthful switching from $\Algi{i_{x-1}}$ to $\Algi{i_{x}}$ as in Section~\ref{sec:nc-truth}. Also, Theorem \ref{thm:truthfulswitching} and the definition of truthful restart in Definition~\ref{def:random-restart} guarantee that our scheduling mechanism will be synced with the restarting algorithms $\xoverline{\Algi{i_{x}}}$ after $\lmax+\dmax$ units of time, and therefore the mechanism generates a welfare that is at least the welfare generated by the new arm in the interval $[\min(t_{x-1}+\lmax+\dmax,t_x-1)+1,t_{x}-1]$. Formally speaking, 
\begin{equation*}
\forall x\in[\hat{T}]:~~\sum_{t=t_{x-1}}^{t_x}\W{t}{\ftbs} \geq \Rit{i_x}{x},\quad\quad\sum_{t=t_{\hat{T}}}^{T}\W{t}{\ftbs} \geq \Rit{i_{\hat{T}}}{\hat{T}}
\end{equation*}
and therefore, by summing over all intervals $\{[t_{x-1},t_{x}]\}_{x\in[\hat{T}]}\cup [t_{\hat{T}},T]$, we have:
\begin{equation}
\label{eq:nclr-lower}
\W{}{\ftbs} \geq \sum_{x\in[\hat{T}+1]} \Rit{i_x}{x}
\end{equation}
Conditioned on the sequence of coin flips $\mathbf{K}[T]$ (and therefore times $t_1,t_2,\ldots,t_{\hat{T}}$), the rewards $\{\Rit{i}{x}\}_{x\in[\hat{T}+1],i\in[n]}$ define an (oblivious) adversarial instance of a MAB problem. For this adversarial instance, time horizon is indeed the random variable $(\hat{T}+1)$, i.e. number of $\heads$ in the sequence $\mathbf{K}[T]$ plus one. Moreover, if we let $\hat{U}=\max_{j\in[\hat{T}+1]}U_j$ where 
\[
j\in[1,\hat{T}+1]: U_j\triangleq t_j-t_{j-1}
\]
then the rewards of this instance would also be upper-bounded by the random variable $C_2\hat{U}$. Then, due to the optimal regret bound of $\switcherbandit$ (Theorem~\ref{thm:exp3}), we have:
\begin{equation}
\label{eq:nclr-regret}
\forall i\in[n]: \expect{\sum_{x\in[\hat{T}+1]}\Rit{i}{x}\lvert~ \mathbf{K}[T]}- \expect{\sum_{x\in[\hat{T}+1]}\Rit{i_x}{x}\lvert~ \mathbf{K}[T]}\leq O(\hat{U}\hat{T}^{1/2}C_2(n\log n)^{1/2})
\end{equation} 
Combining Inequalities (\ref{eq:nclr-upper}), (\ref{eq:nclr-lower}) and (\ref{eq:nclr-regret}) and taking expectations:
\begin{equation}
\label{eq:nclr-expect-regret}
\reg(\ftbs) \leq O(\expect{\hat{U}\hat{T}^{1/2}C_2(n\log n)^{1/2}})+\expect{C_1(\hat{T}+1)}
\end{equation}
Now, $\expect{C_1(\hat{T}+1)}=O(\gamma \cdot C_1 T)$. To bound the other term, we use the following fact, proved in~\cite{eisenberg2008expectation}, about independent and identically distributed geometric random variables. 
\begin{lemma}[\cite{eisenberg2008expectation}]
\label{lem:geometry2}
If $g_1,\ldots,g_k$ are i.i.d. and $g_i\sim\textrm{Geometric}(\gamma)$, then 
\[\expect{\max_{j\in[k]}g_j}\leq H_{k}\cdot \gamma^{-1}=O(\gamma^{-1}\log k)\]. 
\end{lemma}
Note that conditioned on $\hat{T}$, random variables $\{U_j\}_{j\in\hat{T}+1}$ are $(\hat{T}+1)$ i.i.d. $\textrm{Geometric}(\gamma)$ random variables. Using Lemma~\ref{lem:geometry2} we have:
\[
\expect{\hat{U}\hat{T}^{1/2}|\hat{T}}=\hat{T}^{1/2}\expect{\max_{j\in[\hat{T}]}U_j | \hat{T}}\leq \hat{T}^{1/2}\cdot H_ {\hat{T}}\cdot\gamma^{-1}=\gamma^{-1}\cdot O(\hat{T}^{1/2}\log(\hat{T}))
\]
Now, function $f(x)=x^{1/2}\log(x)$ is concave. By taking expectation and using Jensen's inequality, we further upper-bound this term. Hence:
\[
\expect{\hat{U}\hat{T}^{1/2}}\leq \gamma^{-1}.O(f(\expect{\hat{T}}))=\gamma^{-1/2}\cdot T^{1/2}\log(\gamma T)
\]
So, we can upper-bound the RHS of (\ref{eq:nclr-expect-regret}) by
\[
O(\gamma^{-1/2}T^{1/2}\log(\gamma T)C_2(n \log n)^{1/2})+O(\gamma.C_1 T)
\]
By setting $\gamma=(\frac{C_2}{C_1})^{2/3} T^{-1/3}(n\log n)^{1/3}$, we get the desired regret bound. 
		\end{proof}
\section{Lower Bounds}
\label{sec:lower-bounds}


We first state a lower bound of $\Omega(\sqrt T)$ on the regret for the clairvoyant scheduling problem. The proof of this theorem can be found in Appendix~\ref{sec:app-proofs}. We note that this lower bound can be extended to $\tilde{\Omega}(m \sqrt{T})$ when there are $m$ machines by simply having $m$ copies of the same set of jobs every time. Similarly the lower bound can be made to scale linearly with $\vmax$ and $\sqrt{\dmax}$, matching the upper bound we give in Theorem~\ref{lem:fts-regret}. 

\begin{theorem}
\label{thm:lb-experts}
There exists an instance of the clairvoyant scheduling problem where the regret of any online algorithm relative to the hindsight optimal posted-pricing-FIFO algorithm is $\Omega(\sqrt T)$.
\end{theorem}

Next we show a lower bound of $\tilde{\Omega}(T^{2/3})$ on the regret for the non-clairvoyant setting that matches our upper bound within polylogarithmic factors. Our lower bound follows by  a reduction from the lower bound given in  \cite{Dekel2014} for the multi-armed bandit problem with switching costs.

For  the bandit problem with switching costs with $n$ actions,   \cite{Dekel2014} show that there exists a sequence of loss functions $\ell_1, \ell_2, \ldots \ell_T$, where $\ell_i: [n] \rightarrow [0,1]$, such that any online algorithm incurs a regret of at least $\tilde{\Omega}(n^{1/3} T^{2/3})$.  We use this loss sequence to define an instance of the non-clairvoyant scheduling problem as follows. First we give a reduction to regret against the welfare benchmark without random restarts. Later we show how to extend the lower bound to apply against the random restart benchmark.

\begin{figure}[h]
	\centering
	\includegraphics[width=1.0\textwidth]{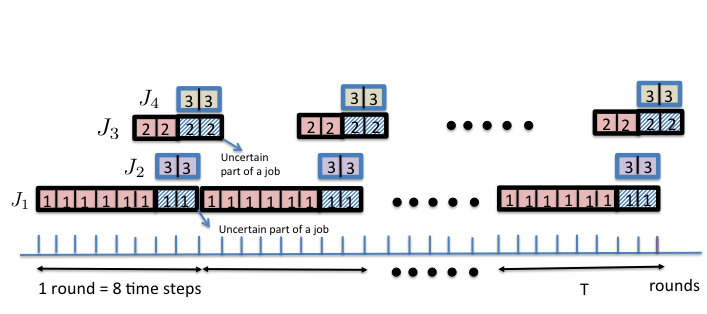}
\caption{ Non-clairvoyant lower-bound instance }
	\label{fig:lowerbound}
\end{figure}
In our lower bound instance, we fix $n= 2$, and let $\ell_i(1)$ and $\ell_i(2)$ denote the losses of actions 1 and 2 in round $i$ as defined in \cite{Dekel2014}.  We map each round of the game to 8 time steps; that is, round $i$ corresponds to the time interval $[8i, 8(i+1)-1]$.  Our instance has 4 sets of jobs $J_1$, $J_2$, $J_3$, and $J_4$, as shown in Figure \ref{fig:lowerbound}. 
In each round, one job from each set  arrives.
Jobs in the set $J_1$ arrive at the beginning of each round; that is, at time steps $8i$ for $i = 0, 1,2, \ldots ..T$. The processing length of a job $j \in J_1$ that arrives in the round $i$ is 6 with probability $p_i(1)$ and 8 with probability $(1- p_t(1))$, where $p_i(1) = 1/2 +  \ell_i(1)/2$.  Observe that processing lengths of the jobs in $J_1$ depend on losses defined by \cite{Dekel2014}. Further, the jobs in $J_1$ have a value of 1 per unit length.  In round $i$, a job from  $J_2$ arrives at time $8(i+1)-2$ for $i = 0, 1, 2, 3, \ldots ..T$, and has a processing length of 2.  The value per unit length of jobs in $J_2$ is 3. 

The set $J_3$ consists of jobs that arrive at time instants $8i-3$ for $i = 1, 2, \ldots T$,  and have value per unit length of 2. The processing length of job $j  \in J_3$ released in the round $i$ is 4 with probability $p_i(2)$ and 2 with probability $1 - p_i(2)$, where $p_i(2) = \ell_i(2)$.  Similar to the jobs in $J_1$, the processing lengths of jobs in $J_3$ depend on the losses defined by the result of \cite{Dekel2014}. Finally, the jobs in set $J_4$ are released at time steps $8i - 1$ for $i = 1, 2, \ldots T$, and have a processing length of 2 and value per unit length of 3.

The deadlines of all jobs are equal to their arrival times. (This condition is not necessary but simplifies the construction.) Hence, if a job is not scheduled upon its arrival, the algorithm loses the job.

Let $\Algi{1}$ and $\Algi{2}$ denote the two  posted price scheduling mechanisms  with prices 1 and 2 and  FIFO scheduling policy. (Note that other posted price mechanisms, for example one that posts a price of 3, lose a constant factor in each round, hence we do not consider them.)
The algorithms  $\Algi{1}$ and $\Algi{2}$ correspond to the action set in the bandit problem with switching costs.   The following two lemmas follow from the construction of lower bound instance. See Appendix~\ref{sec:app-proofs} for a proof.

\begin{lemma}
	\label{lem:reward}
For all rounds $i = 0,1, \ldots T$, the expected value of $\Algi{1}$ in round $i$ is $10 - 2 \ell_i(1)$ and expected value of $\Algi{2}$ in round $i$ is $10 - 2 \ell_i(2)$.
\end{lemma}


\begin{lemma}
		\label{lem:switch}
If an online algorithm switches from $\Algi{2}$ to $\Algi{1}$ in any round $i$, it incurs a loss of at least 6.
\end{lemma}

From Lemmas \ref{lem:reward} and \ref{lem:switch}, we get the following theorem.

\begin{theorem}
\label{thm:lbbandit}
Minimax regret of non-clairvoyant scheduling problem is at least $\tilde{\Omega}(T^{2/3})$.
\end{theorem}

\begin{proof}
Consider an online scheduling algorithm for the non-clairvoyant scheduling problem. At the beginning of each round $i$, it can either follow $\Algi{1}$ and obtain a value of $10 - 2 \ell_i(1)$ or follow $\Algi{2}$ and get a value $10 - 2 \ell_i(2)$. Furthermore, since we are in the non-clairvoyant setting, the online algorithm won't know the value it can obtain from the algorithm it is not following.  Since, the value obtained by $\Algi{1}$ and $\Algi{2}$ are exactly the same in each round except for the terms $-\ell_i(1)$ and $-\ell_i(2)$, the regret of the non-clairvoyant scheduling algorithm is equal to the regret it suffers on the losses at each round $i$. Moreover, switching from $\Algi{2}$ to $\Algi{1}$ in any round incurs a cost of 6.
Hence, our scheduling instance corresponds to bandit with switching cost problem, where switching from  action 2  to action 1 incurs a cost. 
The result of  \cite{Dekel2014}  shows that the problem has a minimax regret of least $\tilde{\Omega}(T^{2/3})$, when there is a switching cost between any pair of actions. However, it is easy to modify the proof in   \cite{Dekel2014}, where there is  a switching cost only between  action 2 to action 1, losing a factor of 2 in the regret bound \cite{Dekel2014}. This completes our reduction.
\end{proof}

 Now we argue that Theorem (\ref{thm:lbbandit}) can be extended to random restarting benchmarks. Before that we make the following simple observation regarding the lower bound instance in Theorem (\ref{thm:lbbandit}). Consider the posted price mechanisms $\Algi{1}$ and $\Algi{2}$ as defined in the proof of Theorem (\ref{thm:lbbandit}). $\xoverline{\Algi{1}} (\gamma),{\xoverline{\Algi{2}}} (\gamma)$ denote the random restarting versions of them with restart parameter $\gamma$. 

 
\begin{lemma}
\label{lem:restart}
For some constant $ c \geq 11$, we have:

\[
j\in\{1,2\}:~~\text{\bf E}[W(\xoverline{\Algi{j}} (\gamma), J)] \geq W(\Algi{j}, J) - c \cdot \gamma \cdot T\] 

\end{lemma} 

\begin{proof}
Proof of the lemma follows immediately from the observation that if $\Algi{j}$ for $j = 1, 2$ restarts at time $t$, and $t$ is in  round $i$, then it loses values of jobs in that round.  The expected number of restarts by the algorithms is at most $\gamma \cdot T$. As the total value of jobs in each round is at most 11,  the statement of the lemma holds for $c \geq 11$.
\end{proof}

To extend the lower bound to random restarting benchmarks, we need the following theorem from \cite{Dekel2014} for the bandit with switching costs problem.

\begin{theorem}[\cite{Dekel2014}]
\label{thm:dekel}
Let $\mathcal{A}$ be a multi-armed bandit algorithm that guarantees an expected regret (without switching costs) of $O(T^\alpha)$ then there exists a sequence of loss functions that forces $\mathcal{A}$ to make $\tilde{\Omega}(T^{2(1-\alpha)})$ switches.
\end{theorem}

Combining Lemma \ref{lem:restart} and above theorem, we prove the following theorem.

\begin{theorem}
\label{thm:lbrestart}
No online algorithm can achieve a regret  $O(T^{2/3 - \epsilon})$ for any $\epsilon > 0$ against a random restarting benchmark with restart parameter $\gamma \in (T^{-1}, T^{-1/3}]$.
\end{theorem}

\begin{proof}
Proof is by contradiction. Suppose there is an online non-clairvoyant algorithm $\mathcal{A}$ that achieves a regret of $O(T^{2/3 - \epsilon})$ for some $\epsilon > 0$. From Lemma \ref{lem:restart}, this implies that it achieves a regret of at most $O(T^{2/3-\epsilon})$ against the non-restarting benchmark of Theorem \ref{thm:lbbandit}. Our proof of Theorem \ref{thm:lbbandit} gives a reduction from the bandit with switching cost problem to the non-clairvoyant scheduling problem. Therefore, we can invoke  Theorem \ref{thm:dekel}, which implies that $\mathcal{A}$ does at least $\tilde{\Omega}(T^{2(1-(2/3 - \epsilon))})$ switches. Since there are only two actions in our lower bound instance, $\mathcal{A}$  still pays a switching cost of at least $\tilde{\Omega}(T^{2(1-(2/3 - \epsilon))})=\tilde{\Omega}(T^{2/3 + 2\epsilon})$ in switching from $\Algi{2}$ to $\Algi{1}$. This is gives a contradiction to the regret of $\mathcal{A}$ being $O(T^{2/3-\epsilon})$ against the non-restarting benchmark, and completes the proof.
\end{proof}

\bibliographystyle{plainnat}
\bibliography{queueing,auctions}
\newpage
\appendix
\section{Deferred proofs for upper bound constructions}

\label{sec:app-proofs-ub}

\begin{numberedlemma}{\ref{lem:truth-restart}}
	{If $A$ is an order respecting truthful mechanism, then $A$ with (arbitrary) restarts is also order respecting and truthful. }
\end{numberedlemma}
\begin{proof}
	We assume that $A$ is restarted in the interval $[1,\dlmax]$, and that the decision to restart is not dependent on what the jobs report.
	For a job that arrives by time 0, reporting an arrival time after 0 is not beneficial since that would mean that this job can only be processed beginning time $\dlmax+1$, and the job would be past its deadline by then. 
	If the job reports an arrival time before time 1, then truthfulness of $A$ guarantees that no misreport is beneficial. 
	
	Now consider a job $j$ that arrives after time 0. 
	If the deadline of this job is such that it has to start by time $\dlmax$, then no matter what it reports it does not get scheduled. 
	For all other jobs, consider the instance where the jobs that arrive during the interval $[1,\dlmax]$ actually arrive at time $\dlmax+1$, with the same arrival order. 
	Truthfulness of $A$ for this instance guarantees truthfulness for such jobs. 
\end{proof} 

\begin{numberedtheorem} {\ref{thm:truthfulswitching}}
	Given \orderresp truthful mechanisms $A$ and $B$ in the non-clairvoyant setting, 
	switching mechanism $C$ in Definition~\ref{def:nclr-truthful-switcher} is \orderresp, truthful, and obtains welfare at least 
	$$ \textstyle \sum_{t\le \dlmax} W_t(A) + \sum_{t\geq1 + \dlmax} W_t(B) ,$$ 
	given that $A$ and $B$ are restarted at time 1. 
\end{numberedtheorem} 
\begin{proof}
	The mechanism $C$ is \orderresp by definition. 
	Truthfulness of mechanism $C$ follows essentially from the truthfulness of restarts (Lemma~\ref{lem:truth-restart}). 
	For any job that arrives before time 1, reporting an arrival time $\geq 1$ is not beneficial since it would only begin processing after $\dlmax$ by which time its deadline would  have passed. 
	If it reports a time before 1, then the truthfulness of $A$ guarantees that no misreport can be beneficial. 
	For any job that arrives after time 0, the situation is exactly the same as a restart. What mechanism was run before time 1 has no bearing on the allocation and payments of this job.
	Lemma~\ref{lem:truth-restart} guarantees truthfulness for these jobs. 
	
	The welfare guarantee follows from observing that mechanism $C$ completes all jobs that mechanism $A$ started  before time 1, and 
	that after time $\dlmax$, the welfare of $C$ matches that of  $B$ due to Lemma \ref{lem:state-match}. 
	The only loss of welfare is in the interval $[1,\dlmax]$ which is bounded by $\vmax(\dlmax)$. 	
\end{proof} 
\section{Doubling Trick In the Multi-Armed Bandit Problem}
\label{appendix:doubling-trick}
In this section, we briefly explain the \emph{doubling trick}. This helps with the case when algorithm $\switcherbandit_1(R,T)$ (that admits a regret guarantee of $O(R\sqrt{T\log(T)})$)  needs to know $R$ and $T$ in advance, and now we want to design an algorithm $\switcherbandit_2$ that does not need knowing these parameters and still wants to achieve $O(R\sqrt{T\log(T)})$ regret bound (assuming $R$ and $T$ are finite). Doubling trick for time horizon $T$ is fairly standard, e.g. see~\cite{auer2007improved}. We show how doubling trick works for range $R$, and then how to merge the two doubling tricks. 
\paragraph{Doubling trick for range $R$.} For simplicity, suppose $R=2^K$ for some integer $K$. $\switcherbandit_2$ does the following. It starts with a guess (initialized to $1$) for $R$ and simulates $\switcherbandit_1$ with this guess. Every time it sees a reward that is not in the guessed range, it doubles the guess (it may double it many times at the same time instance) and starts from scratch. Suppose $T_1,T_2,\ldots,T_k$ are the length of time intervals between two doubling. Therefore, for some constant $c>0$, 
\[
\reg(\switcherbandit_2)\leq \sum_{j=1}^{K}c.2^j\sqrt{T_j\log(T_j)}\leq \sqrt{T\log(T)}\sum_{j=1}^{K}c.2^j=O(R\sqrt{T\log(T)})
\]
\paragraph{Doubling trick for both range $R$ and time horizon $T$.} In order to do so, use the doubling trick for $R$ as a black-box, and during the steps of doubling trick for $T$, use this black-box.

\section{Deferred proofs for lower bound constructions}

\label{sec:app-proofs}
\subsection{Lower bound of $\Omega(\sqrt{T})$  for the clairvoyant setting}
\begin{figure}[h]
	\centering
	\includegraphics[width=0.8\textwidth]{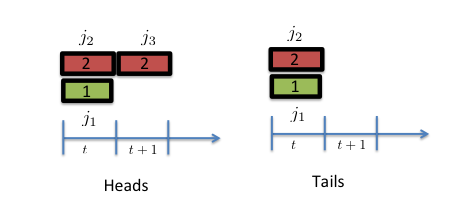}
		\caption{Clairvoyant lower-bound instance}
	\label{fig:lowerboundexperts}
\end{figure}

\begin{numberedtheorem}{\ref{thm:lb-experts}}
There exists an instance of the clairvoyant scheduling problem where the regret of any online algorithm relative to the hindsight optimal posted-pricing-FIFO algorithm is $\Omega(\sqrt T)$.
\end{numberedtheorem}

\begin{proof}
Our proof is an adaptation of the lower bound instance for the problem of prediction with expert advice. 
Recall the lower bound  of $\Omega(\sqrt T)$ for the problem of prediction with expert advice. In this instance, we have two experts. 
In each round, the adversary chooses one of the two experts uniformly at random (and independently of the previous rounds), and assigns a reward of 1. 
The adversary sets the reward of the other expert to zero. 
The expected reward of any online algorithm is $T/2$ but suffers a regret of $\Omega(\sqrt T)$, as the expected reward of the best expert is $\text{\bf E}(\max \{ \# \text{Heads}, \# \text{Tails} \}) =  T/2 + \Omega(\sqrt T)$.
We adapt this lower bound to our scheduling problem as follows.

We map each round of the game to two time steps. We now describe our scheduling instance; we refer the reader to Figure \ref{fig:lowerboundexperts}. Let round $i$ of the game correspond to time steps $[t, t+1]$. In our instance, all jobs have unit processing lengths.
At time $t$, we release two jobs: $j_1$ and $j_2$.  The deadline of job $j_1$ is $t+1$ whereas the deadline of job $j_2$ is $t+2$. Job $j_1$ has a value of 1 and job $j_2$ has a value of 2.  At time $t+1$, we toss an unbiased coin. If the coin lands heads, we release job $j_3$ that has a deadline of $t+2$ and value 2.   We repeat this instance in each round of the game.

Now consider two posted price algorithms $\Algi{1}$ and $\Algi{2}$ with prices 1 and $2$ and FIFO scheduling policy.  It is easy to check that $\Algi{1}$ schedules jobs $j_1$ and $j_2$ in all rounds and obtains a total value of $3T$.   On the other hand,  $\Algi{2}$ processes jobs $j_2$  in all the rounds, and processes jobs $j_3$ if it arrives, depending on the outcome of the coin toss. Therefore, the expected value obtained by $\Algi{2}$ is also $3T$.

Consider the decision of an online algorithm at the beginning  of each round. If it decides to schedule job $j_1$, then the maximum value it can get in a round is equal to that of $\Algi{1}$. On the other hand, if it decides schedule $j_2$, then the maximum value it can get is equal to that of $\Algi{2}$.  Hence, the expected value obtained by any online algorithm is 3T.

Consider the expected value of the better of $\Algi{1}$ and $\Algi{2}$.  Let $\Algi{1}(i)$ and $\Algi{1}(i)$ denote the value obtained by the algorithms in round $i$. Let $X_i$ denote a random variable that takes a value of 1 with probability 1/2 and -1 with probability 1/2.
\begin{eqnarray}
\text{\bf E}( \max \{ \sum_{i} \Algi{1}(i),  \sum_{i} \Algi{1}(2) \} ) &=& \text{\bf E}(\max \{ 3T,  3T + \sum_{i} X_i \}) \nonumber \\
 &=& 3T + \text{\bf E}(\max \{ 0, \sum_{i} X_i \}) \nonumber \\
  &=& 3T + \Omega(\sqrt T) \nonumber
\end{eqnarray}
Therefore, the regret of the online algorithm is at least $\Omega (\sqrt T)$.
\end{proof}

\subsection{ Lower bound for the non-clairvoyant setting with random restarts }

\begin{numberedlemma}{\ref{lem:reward}}
For all rounds $i = 0,1, \ldots T$, the expected value of $\Algi{1}$ in round $i$ is $10 - 2 \ell_i(1)$ and expected value of $\Algi{2}$ in round $i$ is $10 - 2 \ell_i(2)$.
\end{numberedlemma}

\begin{proof}
Consider $\Algi{1}$.  We note that $\Algi{1}$ only schedules jobs from the sets $J_1$ and $J_2$. This is because, whenever $\Algi{1}$ finishes processing a job, there is a job belonging to $J_1$ or $J_2$ that is released exactly at that time and no jobs belonging to $J_3$ and $J_4$ are available for processing. Moreover, $\Algi{1}$ processes jobs from the set $J_1$ in all rounds, and it processes the job from $J_3$ if the processing length of job from $J_1$ is 6. This follows from our construction where the completion time of the job from $J_1$ with processing length 6 coincides exactly with the release time of the job from $J_2$. Therefore, the expected reward obtained by $\Algi{1}$ in round $i$ is $6 + 2 \cdot p_i(1) + (1 - p_i(2)  ) 2 \cdot 3 =  10 - 2 \ell_i(1) $. 

Similarly, it is easy to check that $\Algi{2}$ only schedules jobs from the sets $J_3$ and $J_4$. Therefore, the expected reward obtained by $\Algi{2}$ in round $i$ is $4 + 4 \cdot p_i(2) + (1 - p_i(2)  ) 2 \cdot 3 =  10 - 2 \ell_i(2) $. 
\end{proof}

\begin{numberedlemma}{\ref{lem:switch}}
If an online algorithm switches from $\Algi{2}$ to $\Algi{1}$ in any round $i$, it incurs a loss of at least 6.
\end{numberedlemma}

\begin{proof}
We can assume  that the switching times of the online algorithm correspond to the completion time of some job or idle periods. If the algorithm switches in the middle of processing a job, the  lemma follows trivially since the value of every job in our instance is at least 6.  Fix a round i, and consider a time instant $t$ when the algorithm switches from $\Algi{2}$ to $\Algi{1}$. We consider two cases.

Case 1: Switching time $t$ corresponds to the completion time of a job from the set $J_4$. In this case, 
$\Algi{1}$ is already processing a job from the set $J_1$, and since deadlines of the jobs are same as arrival times, the online algorithm will not be able to schedule the job from the set $J_1$. Therefore, the online algorithm loses a job of value of at least 6 in this round, which we charge to the switching cost.

Case 2: Switching time $t$ corresponds to the completion time of a job from the set $J_3$.  This case has two sub-cases depending on whether the job from the set $J_3$ had a processing length of 2 or 4. Suppose processing length of the job was 2. In this scenario, the online algorithm will not be able to process the job from the set $J_4$, and hence loses a value of 6 in round $i$. On the other hand,  if the processing length of the job from the set $J_3$ was 4, then, the algorithm loses the job $j$ from set $J_1$ in the round $i+1$, since the time of switch $t$ is greater than the  deadline $d_j$ of the job.

In the both cases, the online algorithm loses a value of 6 in round $i$, which we charge to the switching cost.
\end{proof}

\subsection{Linear lower bound for non-clairvoyant setting without random restarts, continuous time}

Here we show that without our assumption regarding random restarting benchmark, no algorithm can get a sub-linear regret for the non-clairvoyant case,  in a continuous time setting. 
By this we mean that the arrival time, deadline  and processing lengths could be real numbers rather than integers. 

\begin{figure}[h]
	\centering
	\includegraphics[width=0.9\textwidth]{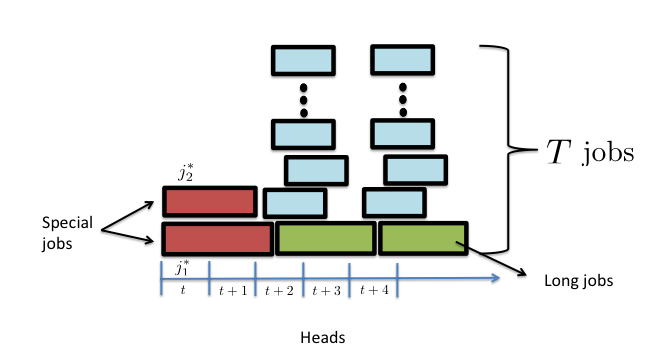}
		\caption{Non-clairvoyant without random restart lower-bound}
	\label{fig:non}
\end{figure}

\begin{theorem}
\label{thm:necrestart}
The minimax regret for the non-clairvoyant scheduling problem in a continuous time setting, compared against a benchmark without  the random restart, is at least  $\Omega(T)$.
\end{theorem}

\begin{proof}
By Yao's minimax principle \cite{Motwani}, we assume that our algorithm $\mathcal{A}$ is deterministic. We now give a distribution over jobs which proves the theorem. 

We map each round of the game into two time steps. Similar to previous lower bound constructions, let $\Algi{1}$ and $\Algi{2}$ denote the two  posted price scheduling mechanisms  with prices 1 and 2 and  FIFO scheduling policy.  All jobs in our lower bound instance have tight deadlines.  See Figure \ref{fig:non} for an illustration of the lower bound instance.

In the first round, the adversary releases two special jobs $j^*_1$ and $j^*_2$.  Job $j^*_1$ has a value 1 and job $j^*_2$ has a value 2; $j^*_1$ arrives slightly earlier than $j^*_2$.  From the definition, $\Algi{1}$ schedules $j^*_1$ and  $\Algi{2}$ schedules $j^*_2$. At the beginning of the game, the adversary tosses an unbiased coin. If the coin comes up heads, then adversary releases a set of jobs $S_1$; otherwise, a set of jobs $S_2$ is released.

The set $S_1$ consists of $T$ jobs, where $T$ is the number of rounds of the game. The value per unit  length of each job is $2$. The release time of a job $j \in S_1$ is $t + X_j$, where $t$ is the beginning of round $i$, and $X_j$ is a uniform $(0,1)$ random variable.  The adversary chooses one of the 
$T$ jobs $j' \in S_1$ uniformly at random, and makes it a {\em long job} by setting the processing length of the job $\ell_{j'} = 2$.  The remaining jobs in $S_1$ have length 1. Furthermore, the adversary correlates the processing length of the special job $j^*_1$ with the release time of the job $j'$, and sets it equal to $2 + X_{j'}$. In other words, the completion time of the special job $j^*_1$ exactly coincides with the release time of $j'$.   The adversary releases $S_1$ at the beginning of each round. Note that the random variables $X_j$ and the long job $j'$ are sampled only in the beginning of game, and do not change over the course of the game.

From the construction it is easy to see that $\Algi{1}$ only schedules long jobs from the set $S_1$, since the completion of the special job $j^*_1$ coincides with the release of the long job, and this repeats in each round till the game ends. The total value obtained by $\Algi{1}$ for this case is $2T$ since it schedules a long job in each round.

The set $S_2$ is constructed exactly same as the set  $S_1$, except for the following differences: 1) The value per unit  length of each job is $2$, and 2) the adversary correlates processing length of the special job $j^*_2$ with release time of the job $j' \in S_2$, and sets it equal to $2 + Y_{j'}$, where $j'$ denotes the long job in set $S_2$. Also notice that random variables $Y_j$ for the set $S_2$ are different from the set $S_1$, but are $(0,1)$ uniform random variables. Again from the construction it is easy to check that $\Algi{2}$ only schedules long jobs, since completion of the job $j^*_2$ coincides with the release of the long job, and this repeats in every round till the game ends. If the coin comes up tails, then the total value obtained by $\Algi{2}$ is $4T$ since it schedules a long job in each round.

Therefore, the expected value of the benchmark is at least $1/2 \cdot 4T+ 1/2 \cdot 4T = 4T$.

Now let us analyze the value obtained by an online algorithm $\mathcal{A}$.  For any realization of the job arrivals,  only one of the two algorithms $\Algi{1}$ or $\Algi{2}$  generate non-trivial value. Since, this is chosen uniformly at random, we conclude that with probability half  $\mathcal{A}$ follows the wrong algorithm in the first round. Therefore, if it needs to achieve a sub-linear regret it has to switch from the algorithm it followed in the first round to the other algorithm. Let us focus on  the case when $\mathcal{A}$ switches from $\Algi{2}$ to $\Algi{1}$ at time $t$.

Now, consider the situation faced by  $\mathcal{A}$ at the beginning of the round $i$ that follows time $t$. It sees $T$ jobs, each with value per unit length of 1, and arriving in the interval $[i, i+1]$ uniformly at random. One of these jobs is a long job, but $\mathcal{A}$  cannot distinguish this since we are in non-clairvoyant setting and the long job is chosen uniformly at random. Therefore, in expectation it takes $T/2$ rounds to identify the long job $\mathcal{A}$. Hence, $\mathcal{A}$ schedules small jobs in $T/2$ rounds. This implies that in expectation $\mathcal{A}$ can only get a value of $T/2 \cdot 1 + T/2 \cdot 4 = 2.5 T$. Similarly, it is easy to argue that if  $\mathcal{A}$ switches from $\Algi{1}$ to $\Algi{2}$ it can get at most  $T/2 \cdot 1 + T/2 \cdot 4 = 2.5T$. Therefore, the expected value obtained  $\mathcal{A}$ on this distribution of jobs is at most $2.5 T$. 

Therefore, it suffers a regret of $ 1.5 T$,  and this completes our proof.

\end{proof}

\end{document}